\documentclass[]{interact}

\usepackage{epstopdf}
\usepackage[caption=false]{subfig}

\usepackage[longnamesfirst,sort]{natbib}
\bibpunct[, ]{(}{)}{;}{a}{,}{,}

\theoremstyle{plain}
\newtheorem{theorem}{Theorem}[section]
\newtheorem{lemma}[theorem]{Lemma}

\theoremstyle{definition}

\theoremstyle{remark}

\begin{document}
\articletype{MANUSCRIPT}

\title{Regression estimator for the tail index}

\author{
	\name{L\'aszl\'o N\'emeth\textsuperscript{a}\thanks{CONTACT L. N\'emeth. Email: lnemeth@caesar.elte.hu, zempleni@ludens.elte.hu} and Andr\'as Zempl\'eni\textsuperscript{b}}
	\affil{\textsuperscript{a,b}E\"otv\"os Lor\'and University,  Institute of Mathematics, Department of Probability and Statistics, Budapest, Hungary.}
}

\maketitle

\noindent\makebox[\linewidth]{\rule{\paperwidth}{0.4pt}}

\begin{abstract}

Estimating the tail index parameter is one of the primal objectives in extreme value theory. For heavy-tailed distributions the Hill estimator is the most popular way to estimate the tail index parameter. Improving the Hill estimator was aimed by recent works with different methods, for example by using bootstrap, or Kolmogorov-Smirnov metric. These methods are asymptotically consistent, but for  tail index $\xi >1$ and smaller sample sizes the estimation fails to approach the theoretical value
for realistic sample sizes. In this paper, we introduce a new empirical method, which can estimate 
high tail index parameters well
and might also be useful for relatively small sample sizes.
\\[1mm]
\end{abstract}


\noindent\makebox[\linewidth]{\rule{\paperwidth}{0.4pt}}


\section{Introduction}

In probability theory and statistics there are many applications, where 
it is essential to know the high quantiles of a distribution, for example solvency margin calculations for insurances or estimating the highest possible loss caused by a natural disaster within a given time period. Extreme value theory provides tools to solve these types of problems. In the 1920's \cite{ft} described the limit behaviour of the maximum of i.i.d. samples. Their theorem is the basis of every research in extreme value theory. Later another approach emerged, where the extremal model is based on the values over a high threshold. Its theoretical background was developed by \cite{bh} and \cite{pi}, the statistical applications by \cite{par}, among others -- summarized by \cite{pot}. Both approaches depend on the tail behaviour of the underlying distribution, which can be measured by the tail index. \cite{hill} constructed an estimator for the tail index, using the largest values of the ordered sample. This estimator is still popular, however finding the optimal number of sample elements to be used remains a challenge. Numerous methods were developed to find the best threshold, for example the double bootstrap method by \cite{dboot}, improved by \cite{boot} or a model based on Kolmogorov-Smirnov distance by \cite{kol}. The Hill estimator is asymptotically consistent for both threshold selection methods, however simulations show that for some sample distributions we need more than $10\,000$ observations for a reasonably accurate estimation. Our new method skips the direct threshold selection for the initial sample and calculates the tail index estimation via the tail indices of simulated subsamples. In this way it results in acceptable estimators for smaller samples, like $n \in (500, 2000)$ too.

\subsection{Mathematical overview}
Let $X_1, X_2, \dots, X_n, \dots$ be an independent and identically distributed (i.i.d) sample from a distribution function $F$, and $M_n=\max(X_1,X_2, \dots, X_n)$. If there exist sequences $a_n >0$ and $b_n \in \mathbb{R}$, such that $$\mathbb{P}\bigg(\frac{M_n-b_n}{a_n}<x\bigg)\to G(x)$$
if $n \to \infty$ for a nondegenerate distribution function $G$, then one can say that $F$ is in the maximum domain of attraction of $G$. The Fisher-Tippet theorem claims that $G$ belongs to a parametric family (with location, scale and shape parameter) called generalized extreme value distribution. For every distribution $G$ exist $a > 0$ and $b$, such that $G^{*}(x)=G(ax+b)$ for every $x$, where $G^{*}$ is the standardized extreme value distribution:

\begin{equation}\label{eq:gev}
G^{*}_{\xi}(x)=\begin{cases}
\exp\{-(1+\xi x)^{-\frac{1}{\xi}}\}, & \text{if\ } \xi \ne 0, \\
\exp\{-e^{-x}\}, & \text{if\ } \xi = 0,
\end{cases}
\end{equation}
where $1+\xi x >0$ holds. The $\xi$ parameter is called the tail index of the distribution. In case of a generalized extreme value distribution the tail index parameter is the shape parameter, which is invariant of standardizing the distribution.

Another approach to the investigation of the tail behaviour is the peaks over threshold (POT) model of \cite{bh} and \cite{pi} where the extremal model is based on the values over a threshold $u$. Let $x_F$ be the right endpoint of the distribution $F$ (finite or infinite).
If the distribution of the standardized excesses over the threshold has a limit, that must be the generalized Pareto distribution:
$$\mathbb{P}\bigg(\frac{X-u}{\sigma_u}<x|X>u\bigg) \to F_{\xi}(x),\quad \text{if\ } u \to x_F,$$

where 
\begin{equation}\label{eq:pot}
F_{\xi}(x)=\begin{cases}
1-(1+\xi x)^{-\frac{1}{\xi}}, & \text{if\ } \xi \ne 0, \\
1-e^x, & \text{if\ } \xi = 0.
\end{cases}
\end{equation}
For the given initial distribution $F$, the two model result in the same parameter $\xi$ in equations (\ref{eq:gev}) and (\ref{eq:pot}). 

 A function $\ell(x)$ is called slowly varying if $\lim_{t\to \infty}\ell(tx)/\ell(t)=1$ for all $x>0$. For tail index parameter  $\xi >0$ the previous limit theorems are true if 

$$1-F(x) = x^{-\frac{1}{\xi}}\ell(x),\quad \text{for $x>0$}.$$

Finding a proper function  $\ell(x)$ and calculating the parameters of the limiting distribution is only applicable for special known distributions. In case of real life problems finding $\ell(x)$ is unrealistic, therefore one can use estimators to approximate $\xi$.

\section{Methods for defining the threshold in Hill estimator}
For tail index $\xi >0$ \cite{hill} proposed an estimator as follows: let $X_1,X_2, \dots, X_n$ be a sample from a distribution function $F$ and $X_1^*\le X_2^*\le \dots \le X_n^*$ the ordered statistic. The Hill estimator for the tail index is
$$\hat{\xi}=\frac{1}{k}\sum_{i=1}^{k} \log X_{n-i+1}^*-\log X_{n-k}^*.$$
Similarly to the POT model, the Hill estimator also uses the largest values of the sample. The threshold is defined as the $(k+1)^{\rm{th}}$ highest observation.

The Hill estimator strongly depends on the choice for $k$. It is important to mention that $\hat{\xi}$ is a consistent estimator for the tail index only if $k \to \infty \text{ and } k/n \to 0$ as $n \to \infty$. If one uses a too small $k$, the estimator has large variance, however for too large $k$, the estimator is likely to be biased.
Therefore proposing a method for choosing the optimal $k$ for the Hill estimator has been in the focus of research by \cite{hall} and others since its invention . 

\subsection{Double bootstrap}\label{doubleboot}
One of the most accurate estimators, the double bootstrap method was introduced by \cite{dboot}, improved by \cite{boot}. In this case, we can find a proposed $k$ by minimizing the asymptotic mean square error.
Let 
$$M^{(i)}(n,k)=\frac{1}{k}\sum_{j=1}^{k} (\log X_{n-j+1}^*-\log X_{n-k}^*)^{i}.$$
One can see, that $\hat{\xi}=M^{(1)}(n,k)$ is the Hill estimation. Instead of $M^{(1)}(n,k)$, the method optimizes $k$ for the mean square error of $M(n,k)=M^{(2)}(n,k)-2(M^{(1)}(n,k))^2$. Let $k_1$ be the optimal threshold index for $M^{(1)}(n,k)$, and $k_2$ for $M(n,k)$. The 
\begin{equation}\label{eq:kk}
\frac{k_2}{k_1} \sim \bigg(1-\frac{1}{\rho}\bigg)^{1/(1-2\rho)}
\end{equation}
 statement is proven by \cite{dboot}, where $\rho$ is a regularity parameter, which can be estimated in a consistent way. This statement allows to estimate $k_2$ instead of $k_1$ by following the next steps.

\begin{itemize}
\item Choose $\epsilon \in (0,1/2),$ and set $m_1=[n^{1-\epsilon}]$ to ensure consistency if $n \to \infty$. Estimate $E(M(m_1,r)^2|X_1,X_2,\dots,X_n)$ by drawing $m_1$ size bootstrap samples from the empirical distribution function $F_n$ and minimize it in $r$. Denote the minimum by $r_1$.
\item Set $m_2=[m_1^{2}/n]$ and minimize $E(M(m_2,r)^2|X_1,X_2,\dots,X_n)$ the same way as in the first step, let the minimum be $r_2$.
\item Estimate the regularity parameter $\rho$, which is important for further calculations, by $\hat{\rho}=\log(r_1)/(-2\log(m_1)+2\log(r_1))$.
\item Now one can estimate the optimal $k$ using the approximation (\ref{eq:kk}) by $$\hat{k}=\frac{r_1^{2}}{r_2}\bigg(1-\frac{1}{\hat{\rho}}\bigg)^{1/(2\hat{\rho}-1)}.$$
\end{itemize}
 The Hill estimator based on the double bootstrap method provides appropriate tail index estimation, but usually results in a long computation time. For smaller sample sizes the acceptable range is limited to $\xi >0.5$.

\subsection{Kolmogorov-Smirnov distance metric}\label{ksdist}
Another approach is to minimize the distance between the tail of the empirical distribution function and the fitted Pareto distribution with the estimated tail index parameter. One can use the Kolmogorov-Smirnov distance for the quantiles as metric as proposed by \cite{kol}. Assume that 
$$F(x)=1-C x^{-\frac{1}{\xi}}+o(x^{-\frac{1}{\xi}}).$$
The quantile function can be approximated by $$x=\bigg(\frac{P(X>x)}{C}\bigg)^{-\xi}.$$
The probability $P(X > x)$  can be replaced by $j/n$, and $\xi$ is estimated by the Hill estimator for some $k$, moreover $C$ can be estimated by $\frac{k}{n}(X_{n-k+1})^{1/\xi}$. 
Using these substitutions one gets an estimation for the quantiles as a function of $j$ and $k$.
$$q(j,k)=\bigg(\frac{k}{j}(X_{n-k+1})^{1/\hat{\xi}}\bigg)^{\hat{\xi}}.$$ The optimal $k$ for the Hill estimator is chosen as the $k$, which minimizes the distance between the empirical and the calculated quantiles
$$\hat{k}=\arg \min_k \bigg[\sup_{j \in 1,\dots,T}|x_{n-j}-q(j,k)|\bigg],$$ where $T$ sets the fitting threshold (we call it KS threshold).
The advantages of this method are that it is easy to program and its computation time is short. As \cite{kol} mentioned, it is the best performing known method if $0< \xi <0.5$ and also works well for small sample sizes. However, for distributions with  tail index $\xi >0.5 $ this technique results in highly biased estimation.

Generally it can be said, that there are good methods for finding an appropriate $k$ if the tail index parameter is $0<\xi<0.5$, which contains distributions with finite variance. However, these methods usually fail to perform well for the case $\xi >0.5$ (these are distributions with infinite variance).

The Kolmogorov-Smirnov method was introduced and analized by simulations in \cite{kol}, but no theoretical background is available for this method. The  Theorem \ref{ksunbias} states that under some conditions the Kolmogorov-Smirnov technique results in unbiased estimation for the tail index. For the proof we need the following lemma:

\begin{lemma}
	If  $k \to \infty$, then $\frac{1}{k}\log(\prod_{j=1}^{k}\frac{k}{j}) \to 1$.
\end{lemma}
\begin{proof}
	For $k \to \infty$ by using the Stirling formula we have
	\begin{eqnarray*}
		\frac{1}{k}\log\bigg(\prod_{j=1}^{k}\frac{k}{j}\bigg)&=&\frac{1}{k}\log\bigg(\frac{k^k}{k!}\bigg)\sim \frac{1}{k}\log\bigg(\frac{k^k}{\sqrt{2 \pi k}\frac{k^k}{e^k}}\bigg)=  \\
		&=&	\frac{1}{k}\log\bigg(\frac{e^k}{\sqrt{2 \pi k}}\bigg)=\frac{1}{k}\bigg(\log(e^{k})-\log(\sqrt{2\pi k})\bigg) \\
		&=&1-\frac{1}{k}\log(\sqrt{2\pi k})\to 1.
	\end{eqnarray*}
\end{proof}

\begin{theorem}\label{ksunbias}
	The Hill estimatior results in asymptotically unbiased estimation for the tail index using the Kolmogorov-Smirnov $k$-selection technique as $k \to\infty$.
\end{theorem}
\begin{proof}
	Let $X_1,X_2, \dots, X_n$ be our sample, and be $X_1^{*},X_2^{*}, \dots, X_n^{*}$ the ordered sample. 
	Let
	\begin{equation}\label{eq:qfunc}
	q_{n}(j,k)=\bigg(\frac{k}{j}(X_{n-k+1})^{1/\xi}\bigg)^{\xi}
	\end{equation}
	be our unbiased quantile estimator for $\mathbb{E}(X_{n-j+1}^{*})$, similar as one can see \cite{kol}. For the right $n\to \infty$, $k(n)\to \infty$ sequences the distribution will be identical to a Pareto distribution (POT model). Therefore Glivenko-Cantelli type theorems for quantiles provides that the $q$ estimator is consistent, for every $\varepsilon > 0$ there is an $M$ such that if $n>M$ then the Kolmogorov-Smirnov distance will be less than $\varepsilon$.
	
	Choose $k$ which minimizes the Kolmogorov-Smirnov distance of $q$ estimator and $X_j^{*}$ samples and let $\hat{\xi}$ be the Hill estimation using $k$ if $n>M$.
	\begin{eqnarray*}
		\mathbb{E}\big(\hat{\xi})&=&\mathbb{E}\bigg(\frac{1}{k}\sum_{j=1}^{k}(\log(X_{n-j+1}^*)-\log(X_{n-k}^*))\bigg) \\
		&=&\frac{1}{k}\log\bigg(\prod_{j=1}^{k}\mathbb{E}(X_{n-j+1}^*)\bigg)-\log(\mathbb{E}(X_{n-k}^*)) \\
		&\sim&\frac{1}{k}\log\big(\prod_{j=1}^{k}q_n(j,k)\big)-\log(q_n(k+1,k)) \quad \text{approximation by using (\ref{eq:qfunc}}) \\
		&=& \frac{\xi}{k} \log\bigg(\prod_{j=1}^{k}\frac{k}{j}(X_{n-k+1}^{*})^{\frac{1}{\xi}}\bigg)-\xi\cdot \log\bigg(\frac{k}{k+1}(X_{n-k+1}^{*})^{\frac{1}{\xi}}\bigg) \\
		&=&\frac{\xi}{k}\bigg(\log(\prod_{j=1}^{k}\frac{k}{j})+\frac{1}{\xi}\sum_{j=1}^{k}\log(X_{n-k+1}^{*})\bigg)-\xi\bigg(\log(\frac{k}{k+1})+\log((X_{n-k+1}^{*})^{\frac{1}{\xi}})\bigg) \\
		&=&\frac{1}{k}\sum_{j=1}^{k}\log(X_{n-k+1}^{*})-\log(X_{n-k+1}^{*}) +\frac{\xi}{k}\bigg(\log(\prod_{j=1}^{k}\frac{k}{j})\bigg)-\xi\bigg(\log(\frac{k}{k+1})\bigg) \\
		&=& \xi \cdot \bigg(\frac{1}{k}\log(\prod_{j=1}^{k}\frac{k}{j})\bigg)-\xi \bigg(\log(\frac{k}{k+1})\bigg) \to \xi+0
	\end{eqnarray*} 	 
\end{proof}

However $k \to \infty$ is not an evident condition and even if it is realized we cannot say anything about the speed of convergence. This could lead to biased tail index estimation using Kolmogorov-Smirnov method, especially if $\xi$ has high value.

\section{Estimator for heavy-tailed distributions}

The method based on Kolmogorov-Smirnov metric is asymptotically unbiased. According to \cite{kol} the method provides acceptable results for small sample sizes if $\xi <0.5$,
but for tail index $\xi >0.5$ usually samples of more than $10\,000$ elements are necessary to estimate $\xi$ properly.
However, in most real-life applications such a large sample size is not available, therefore in cases of distributions with infinite variance the Kolmogorov-Smirnov metric-based estimator has significant bias.

Independent data in size $500$ from distributions with different tail index parameter $\xi$ were simulated to detect the magnitude of the bias. The estimations were calculated $30\,000$ times for each $\xi$. 
This experiment showed that the distribution of the estimator is not normal, but in spite of the different initial distributions for each $\xi$, we received similar empirical distributions for the estimators, which could be approximated by a generalized extreme value distribution, as one can see in Figure~\ref{fig:kol_smir_density}. Although the GEV fitting is the best among known distributions, the goodness of fit tests still reject it. Therefore it would be beneficial characterizing the bootstrap distribution more precisely. We estimated the parameters of the GEV distribution using maximum likelihood method, which is a consistent estimator, although sometimes requires high sample size for proper estimation.

However, based on the simulations, one can detect a linear relation between the parameters of the best fitting GEV distribution and the theoretical tail index parameter $\xi$ in the interval $ (0.5,4)$ (see Figure~\ref{fig:parameter_regression}). For $\xi <0.5$ the mean of the estimations using the Kolmogorov-Smirnov method was acceptable. For higher index values a linear transformation can be used on the parameters of the fitted GEV distribution to estimate the tail index. When $\xi >4$ the Kolmogorov-Smirnov method could not return reasonable results, so the correction could not help.

Alternatively, the average of estimations and the theoretical tail index parameter are also in linear connection according to simulations. 

After recalculating the mean values and GEV parameters for sample sizes of $2000$, $8000$ and $10\,000$ with $10\,000$, $2000$ and $1000$ simulations we experienced that the values were similar as one can see in Table~\ref{table:fitting}. 
Due to the long calculation time no simulations were run for higher $(\text{more than }10\,000)$ sample sizes, especially that the asymptotic behaviour of other methods are already realized for these sample sizes.
With these experiments we can state that in $\xi \in (0.5,4)$ the distribution of the Kolmogorov-Smirnov estimation does not depend on the sample size if $500<n<8000$. Deviation were detected  only on $\xi\geq 3$, $n=8000$ case. These results allow to use linear regression between the estimated GEV parameters and the theoretical tail index parameter.

\begin{table}[h]
\centering
\caption{The mean and the parameters of the fitted GEV distribution for Kolmogorov-Smirnov estimations}
\label{table:fitting}
\begin{tabular}{||c r|| c| c| c| c||}
\hline
\hline 
 && mean & location & scale & shape \\[0.1em]
 \hline
 \hline
 $\xi=0.5$,& n=500 & 0.448 &0.378 & 0.113 & 0.043 \\
 &n=2000 & 0.446 & 0.378 & 0.112 & 0.032 \\
 &n=8000 & 0.454 & 0.383& 0.116 & 0.045 \\
 &n=10000 &0.448 & 0.382&0.116 &-0.001 \\
 \hline
 $\xi=1$,& n=500 & 0.85 &0.701 & 0.274 & -0.03 \\
 &n=2000 & 0.855 & 0.709 & 0.274 & -0.039 \\
 &n=8000 & 0.859 & 0.724 & 0.27 & -0.078 \\
 &n=10000 &0.856 & 0.713&0.276 &-0.062 \\
 \hline
 $\xi=2$,& n=500 & 1.628 &1.334 & 0.605 & -0.094 \\
 &n=2000 & 1.634 & 1.356 & 0.601 & -0.115 \\
 &n=8000 & 1.617 & 1.336 & 0.602 & -0.108 \\
 &n=10000 &1.59 & 1.316&0.594 &-0.118 \\
 \hline
 $\xi=3$,& n=500 & 2.392 &1.978 & 0.935 & -0.141\\
 &n=2000 & 2.36 & 1.963 & 0.903&-0.137 \\
 &n=8000 &  2.408  & 2.055  & 0.911   & -0.22   \\
 &n=10000 & 2.441&2.09 &0.903 &-0.217 \\
 \hline
 $\xi=4$,& n=500 & 2.962 &2.478 & 1.159 & -0.177\\
  &n=2000  & 3.137 & 2.59 & 1.232&-0.128 \\
  &n=8000 &  3.149  & 2.637  & 1.225   & -0.163   \\
  &n=10000 & 3.09& 2.573 & 1.2 & -0.141 \\
  \hline\hline
 
\end{tabular}
\end{table}

\begin{figure}[h!]
	\centering
	\includegraphics[width=0.55\textwidth]{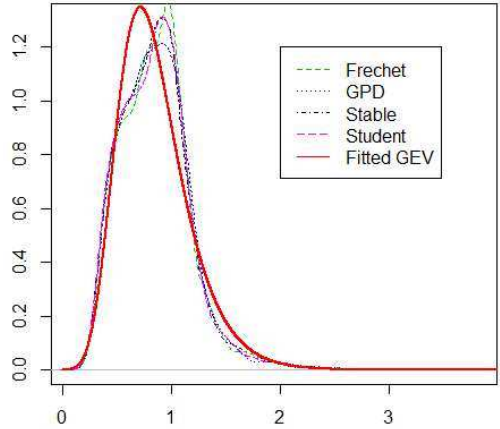} 
	\caption{Empirical density functions of Kolmogorov-Smirnov estimations for  tail index $\xi=1$ based on $1000$ size samples using $5000$ simulations and density function of the fitted GEV distribution (parameters: $\xi=-0.04$, $\mu=0.704$, $\sigma=0.273$) using all of the data                                                                                                                              
	}
	\label{fig:kol_smir_density}
\end{figure}

\begin{figure}[h!]
	\centering
	\includegraphics[width=0.6\textwidth]{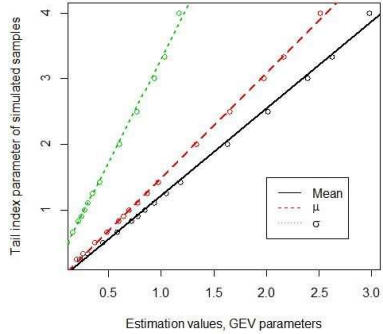} 
	\caption{Linear regression for mean and EVD parameters. Average Kolmogorov-Smirnov estimations were calculated using $500$ size Frechet samples and $30000$ simulations}
	\label{fig:parameter_regression}
\end{figure}

Based on the observations above we constructed an algorithm that can estimate the tail index on the interval $\xi \in (0.5,4)$  for 
sample sizes between $500$ and $10\,000$. 
The connection stands for even smaller samples, however the proper behavior of the bootstrap simulations require this sample size. For smaller sample the information about the extremes is minimal and even using bootstrap resampling results in estimates with high variance. The sample size has no significant effect on the estimation, when modeling is done separately for the samples. If one constructs a model with the mean of the Hill estimators, the location parameter of the fitted GEV and the sample size, then the size becomes marginally significant with a negligible coefficient ($p=0.017$, coef$=-0.0000046$) which has noticeable effect only over $100\,000$ observations.

\subsection{Algorithm}\label{regest}

Let $X_1,X_2,\dots,X_n$ be independent and  identically distributed observations.

\begin{enumerate}
\item As we investigate the extremal behaviour of the data, we apply the $m$ out of $n$ bootstrap method to resample from the observations. 
According to \cite{mout1}, see also \cite{moutn} in case of extreme-value inference  it is important to set $m$ as $m \to \infty \text{ and } m/n \to 0$ as $n \to \infty$. 
One way to ensure this property is if we chose $m=n^{\epsilon}$, where $\epsilon \in (0.5,1)$ can be arbitrary. 

\item Calculate $\xi_1,\xi_2,\dots,\xi_M$ tail index estimations from the $M$ bootstrap samples using the Kolmogorov-Smirnov method (\ref{ksdist}) for Hill estimator, where the bootstrap sample size is $m$. Let $\xi^*=(\sum_{j=1}^{M}{\xi_j})/{M}$ be the average of Kolmogorov-Smirnov estimations.

\item  Fit a generalized extreme value distribution to the estimated values. Our experience is that the tail index of the initial distribution is in linear relation with all of the extreme value parameters. 

The covariance matrix of the best fitted GEV distribution parameters has the eigenvalues $0.666, 0.0005, 0.00006$, which means that there is only one significant factor. The eigenvector of the largest eigenvalue is $(0.89,0.45,-0.09)$, so the location and scale parameters have higher impact. If we apply a linear regression with these parameters only, the location parameter has significant coefficient $(p=0.0003)$. Therefore we suggest to use the location parameter $\mu$ for estimation.

\item Let $\hat{\xi}=-0.119+1.603\cdot \mu$. These values were calculated from the  linear regression of the previous step.
\item Alternatively, we could have used the estimator $\hat{\xi} ^*=-0.1181+1.3301\cdot {\xi} ^*$, where ${\xi} ^*$ is the mean of the bootstrap samples.
\end{enumerate}

\subsection{Simulations}
To examine the properties of the fitting regression estimator, we simulated data from Pareto, Frechet, Student and symmetric stable distributions with the same theoretical tail index parameter for samples of size $500$, $1000$ and $4000$. We compared the average of the calculated values by using our new regression method to the theoretical parameter between $\xi=0.2$ and $\xi=5$. Subsequently we calculated the average absolute error from the theoretical $\xi$. For the different sample sizes we set the bootstrap subsample size to $60$, $100$ or $300$, respectively. We set a suitable KS threshold $(15,30,120)$ in each case, however our previous experience shows that if the KS threshold is higher than $5\%$ of the bootstrap sample, its significance is negligible because the largest deviation is observed for the highest quantiles.

One may conclude by Table \ref{table:estim_500}, \ref{table:estim_1000} and \ref{table:estim_4000} that in the interval $\xi \in (0.5,4)$ the estimates are near to the theoretical value. We can see that the mean of error for smaller tail indices reduces by using higher sample sizes, while the estimator gets biased for larger tail index. For $n=500$ the estimation is still acceptable for $\xi=5$. The only exception is the symmetric stable case for $\xi=0.5$, but this is practically the normal distribution, which is not in the maximum domain of attraction of the Frechet distribution, so we did not expect this extreme value model to work well.

\begin{table}[h]
\centering
\caption{Average tail index estimations from $500$ size samples, by using $500$ simulations, while the bootstrap sample size is $m =60$. The KS threshold is set to $15$. For every distribution we calculated the absolute error from the theoretical parameter $\xi$ }
\label{table:estim_500}
\begin{tabular}{||c|| c| c| c| c| c| c|c| c||}
\hline
\hline 
$\xi$ & 0.2 &0.33 & 0.5 & 1 & 2 & 3 &4& 5 \\[0.1em]
 \hline
 \hline
 GPD & 0.44 &0.52 & 0.63 & 1.02 &1.93& 2.86 &3.83&4.77 \\
 error & 0.24 & 0.19 & 0.13 & 0.11&0.22 & 0.33 &0.43& 0.52 \\
 \hline
 Frechet & 0.13& 0.29 & 0.48 & 0.99 & 1.96&2.92&3.89& 4.8 \\
 error & 0.07 & 0.05 & 0.05& 0.1 & 0.22& 0.33&0.45& 0.53 \\
 \hline
 Student & 0.41&0.48& 0.58&0.98&1.87&2.79&3.77&4.72\\
 error & 0.21&0.15&0.09&0.1& 0.23&0.35&0.47&0.59\\
 \hline
 Stable & & &0.32&0.99&1.98&2.96&3.94&4.91\\
 error& & & 0.18&0.11&0.22&0.33&0.45&0.54 \\
 \hline\hline
\end{tabular}
\end{table}

\begin{table}
\centering
\caption{Average tail index estimations from $1000$ size samples, by using $200$ simulations, while the bootstrap sample size is $100$. The KS threshold is set to $30$. For every distribution we calculated the absolute error from the theoretical parameter $\xi$}
\label{table:estim_1000}
\begin{tabular}{||c|| c| c| c| c| c| c| c|c||}
\hline
\hline
$\xi$ & 0.2 &0.33 & 0.5 & 1 & 2 & 3 &4& 5 \\
 \hline
 \hline
 GPD & 0.41 &0.49 & 0.6 & 1.01 &1.95& 2.92 &3.91&4.84 \\
 error & 0.21 & 0.16 & 0.1 & 0.09&0.18 & 0.27&0.36 & 0.41 \\
 \hline
 Frechet & 0.14& 0.3 & 0.48 & 0.99 & 1.97&2.95&3.93& 4.74 \\
 error & 0.06 & 0.04 & 0.05& 0.09 & 0.19& 0.29&0.4 &0.47 \\
 \hline
 Student & 0.37&0.45& 0.55&0.97&1.91&2.89&3.81&4.83\\
 error & 0.17&0.12&0.06&0.09& 0.21&0.29&0.41&0.47\\
 \hline
 Stable & & &0.28&0.97&1.97&2.95&3.93&4.88\\
 error& & & 0.22&0.09&0.2&0.29&0.39&0.47 \\
 \hline
 \hline
\end{tabular}

\end{table}

\begin{table}
\centering
\caption{Average tail index estimations from $4000$ size samples, by using $200$ simulations, while the bootstrap sample size is $300$. The KS threshold is set to $120$. For every distribution we calculated the absolute error from the theoretical parameter $\xi$}
\label{table:estim_4000}
\begin{tabular}{||c|| c| c| c| c| c| c|c| c||}
\hline
\hline
$\xi$ & 0.2 &0.33 & 0.5 & 1 & 2 & 3 &4& 5 \\
 \hline
 \hline
 GPD & 0.36 &0.43 & 0.55 & 0.99 &1.98& 2.98 &3.97&4.53 \\
 error & 0.16 & 0.1 & 0.06 & 0.09&0.18 & 0.26 &0.34 &0.52 \\
 \hline
 Frechet & 0.15& 0.3 & 0.48 & 1 & 2&3.02&3.96& 4.11 \\
 error & 0.05 & 0.03 & 0.04& 0.08 & 0.17& 0.25&0.29 &0.9 \\
 \hline
 Student & 0.31&0.39& 0.51&0.99&1.97&2.96&3.97&4.85\\
 error & 0.11&0.05&0.03&0.09& 0.19&0.27&0.33&0.39\\
 \hline
 Stable & & &0.22&0.99&1.99&2.98&3.91&4.7\\
 error& & & 0.28&0.08&0.16&0.26&0.31&0.42 \\
 \hline
 \hline
\end{tabular}
\end{table}

\subsection{Comparison of the methods}

As we mentioned the Kolmogorov-Smirnov method (\ref{ksdist}) was developed to estimate tail index parameter in $\xi \in (0,0.5)$ interval. The double bootstrap (\ref{doubleboot}) method can be used for estimating larger values, but it has long computational time. However our new algorithm works well for larger values than the Kolmogorov-Smirnov method  and it is faster than the double  bootstrap. In this section we compare the methods for different $\xi$ by using samples of size $1000$, and present the results in Table \ref{table:comp_frechet},\ref{table:comp_student}, \ref{table:comp_stable} and \ref{table:comp_pareto}.

\begin{table}
\centering
\caption{Comparing the methods for Frechet distribution using $200$ simulations from samples of size $1000$. In double bootstrap method we used $\epsilon=0.15$ for determining the bootstrap sample size ratio and simulated $500$ times in every bootstrap simulation. For fitting and mean regression the subsample size is $100$ and the KS threshold is set to $30$. For every method we calculated the mean absolute error from the theoretical parameter $\xi$}
\label{table:comp_frechet}
\begin{tabular}{||c|| c| c| c| c| c|c|c|c|| }
\hline
\hline
 $\xi$&0.2&0.33& 0.5 &1 &  2 &3& 4 & 5 \\
 \hline
 \hline
 Double bootstrap &0.21&0.35&0.53&1.06&2.11&3.17&4.23&5.29 \\
 error &0.02&0.04&0.05&0.1&0.21&0.31&0.41&0.52 \\
 \hline
Kolmogorov-Smirnov &0.18&0.3&0.44&0.82&1.61&2.35&2.57&1.65 \\
error &0.04&0.08&0.13&0.29&0.58&0.88&1.44&3.35\\
 \hline
Regression fitting &0.14&0.3&0.48&0.99&1.97&2.95&3.93&4.74 \\
error &0.06&0.04&0.05&0.09&0.19&0.29&0.4&0.46 \\

 \hline
 Mean regression &0.16&0.31&0.5&1.02&2&2.97&3.91&4.59 \\
 error &0.04&0.03&0.04&0.09&0.18&0.27&0.35&0.48 \\
 \hline
 \hline
\end{tabular}
\end{table}

\begin{table}
\centering
\caption{Comparing the methods for Student distribution using $200$ simulations from samples of size $1000$. In double bootstrap method we used $\epsilon=0.15$ for determining the bootstrap sample size ratio and simulated $500$ time in every bootstrap simulation. For fitting and mean regression the subsample size is $100$ and the KS threshold is set to $30$. For every method we calculated the absolute error from the theoretical parameter $\xi$}
\label{table:comp_student}
\begin{tabular}{||c|| c| c| c| c| c|c|c|c||| }
\hline
\hline
 $\xi$ &0.2&0.33& 0.5 &1 &  2 &3& 4 & 5 \\
 \hline
 \hline
 Double bootstrap &0.26&0.38&0.55&1.02&2.03&3.06&4.05&5.13 \\
 error &0.1&0.11&0.1&0.11&0.14&0.2&0.25&0.32 \\
 \hline
Kolmogorov-Smirnov &0.25&0.33&0.47&0.83&1.57&2.41&3&3.05 \\
error & 0.06&0.07&0.13&0.31&0.6&0.91&1.19&1.95\\
 \hline
Regression fitting &0.37&0.45&0.55&0.97&1.91&2.89&3.81&4.83 \\
error &0.17&0.11&0.06&0.09&0.21&0.29&0.41&0.47\\
 \hline
 Mean regression &0.36&0.45&0.57&1&1.96&2.97&3.88&4.81 \\
 error &0.16&0.12&0.08&0.09&0.19&0.27&0.37&0.43\\
 \hline
 \hline
\end{tabular}
\end{table}

\begin{table}
\centering
\caption{Comparing the methods for symmetric stable distribution using $200$ simulations from samples of size $1000$. In double bootstrap method we used $\epsilon=0.15$ for determining the bootstrap sample size ratio and simulated $500$ times in every bootstrap simulation. For fitting and mean regression the subsample size is $100$ and the KS threshold is set to $30$. For every method we calculated the mean absolute error from the theoretical parameter $\xi$}
\label{table:comp_stable}
\begin{tabular}{||c|| c| c| c| c|c| c|| }
\hline
\hline
$\xi$ & 0.5 &1 &  2&3 & 4 & 5 \\
 \hline\hline
 Double bootstrap & 0.12 &1.02  & 2.13&3.2 &4.29& 5.37 \\
 error & 0.38&0.11&0.27&0.4&0.49&0.62 \\
 \hline
Kolmogorov-Smirnov & 0.17&0.84& 1.6&2.36 &3.12&2.63 \\
error & 0.33&0.28&0.58& 0.9 &1.18&2.37\\
 \hline
 Regression fitting& 0.28&0.97&1.97& 2.95 &3.93&4.88 \\
error &0.22&0.09&0.2& 0.29 &0.39&0.47 \\
\hline
 Mean regression & 0.27&1.01&2.03& 3.02 &3.99&4.86 \\
 error &0.23&0.09&0.19& 0.28 &0.35&0.41 \\
 \hline\hline
\end{tabular}
\end{table}

\begin{table}
\centering
\caption{Comparing the methods for Pareto distribution using $200$ simulations from samples of size $1000$. In double bootstrap method we used $\epsilon=0.15$ for determining the bootstrap sample size ratio and simulated $500$ times in every bootstrap simulation. For fitting and mean regression the subsample size is $100$ and the KS threshold is set to $30$. For every method we calculated the mean absolute error from the theoretical parameter $\xi$}
\label{table:comp_pareto}
\begin{tabular}{||c|| c| c| c| c| c|c|c|c|| }
\hline\hline
$\xi$ & 0.2& 0.33 &0.5 &1 &  2& 3 & 4 & 5 \\
 \hline\hline
 Double bootstrap & 0.34& 0.44& 0.59 &1.07  & 2.05&3.06 &4.05& 5.04 \\
 error &0.7&0.15 & 0.15&0.12&0.14& 0.16&0.19&0.21 \\
 \hline
Kolmogorov-Smirnov &0.28 &0.36 & 0.48&0.85& 1.61& 2.36&2.77&2.28 \\
error &0.08&0.09 & 0.13&0.28&0.61&0.88&1.31&2.72\\
 \hline
 Regression fitting& 0.41&0.49& 0.6&1.01&1.95&2.92&3.91&4.84 \\
error &0.21&0.16&0.1&0.09&0.18&0.27&0.36&0.41 \\
\hline
 Mean regression &0.41&0.49& 0.62&1.04&1.99&2.95&3.9&4.7 \\
 error &0.21&0.16&0.12&0.09&0.17&0.25&0.33&0.41 \\
 \hline\hline
\end{tabular}
\end{table}

For every distribution we can see, that the double bootstrap and the regression estimators (\ref{regest}) have better properties: lower mean square error than the Kolmogorov-Smirnov method for $\xi >0.5$, but if $\xi <0.5$ the Kolmogorov-Smirnov method starts to perform better. We can conclude that the double bootstrap and regression methods have similar accuracy, but our new method has lower computational time and the estimations 
can usually be approximated by the normal distribution (see Figure \ref{fig:nortest}).

\begin{figure}[h!]
  \centering

        \includegraphics[width=0.6\textwidth]{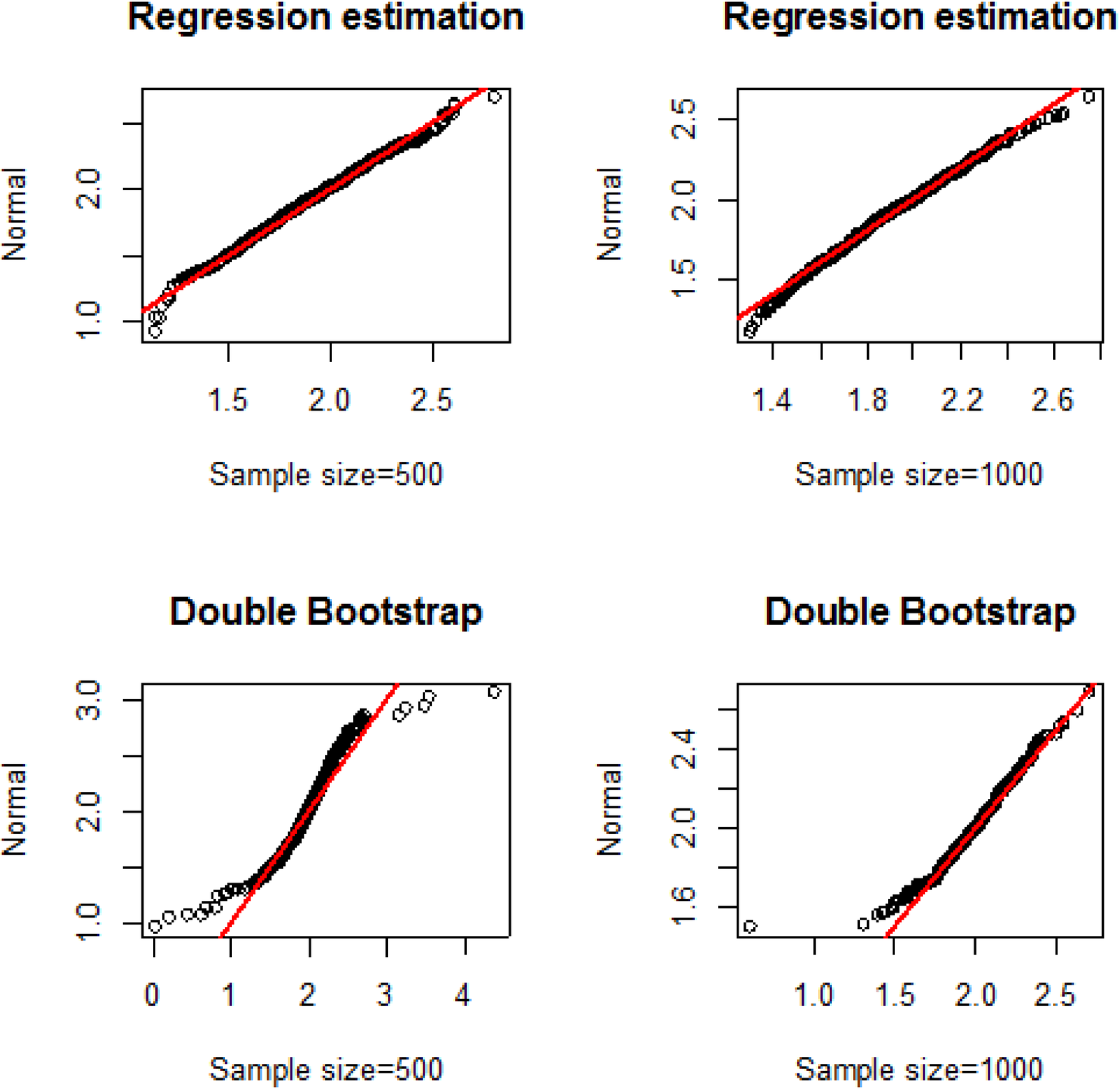} 
      
        \caption{QQ plots for case $\xi=2$ with regression estimation and double bootstrap to sample sizes of $500$ and $1000$. We have got $p$-values of $0.32$ and $0.05$  by the Anderson-Darling test for regression estimation, but for double bootstrap both $p$-values were less than $0.0001$.  }
    	\label{fig:nortest}
\end{figure}
    
In the $\xi \in (0.5,4)$ region the regression estimators have similar errors for the investigated distributions, 
which means it does not depend on the distribution of the sample, unlike the double bootstrap method.

The results of simulations in Table \ref{table:comp_frechet},\ref{table:comp_student}, \ref{table:comp_stable} and \ref{table:comp_pareto} let us provide a model selection method for estimating the tail index. One can choose the best method for the sample by following the next steps:

\begin{enumerate}
	\item Estimate the tail index using the Kolmogorov-Smirnov method (\ref{ksdist}). If the estimated $\hat{\xi}$ value is less than $0.5$, the Kolmogorov-Smirnov method is the best among the investigated aproaches, thus the estimation is acceptable.
	\item If $0.5<\hat{\xi}<0.85$, then it is likely that for the true tail index  $0.5<\xi<1$ holds, thus the regression estimation (\ref{regest}) gives the best estimation.
	\item If $0.85<\hat{\xi}$, then the behavior of the methods depends on the initial distribution. Fit GPD, Student, Stable and Frechet distributions for the sample.
	\item If the best fitting distribution is GPD or Student, use the double bootstrap method (\ref{doubleboot}).
	\item If the best fitting distribution is Stable or Student, use the regression method.
\end{enumerate}

\subsection{Applications to real data}

The Danish fire losses is a well known data set, which is suitable for testing extreme value models. Previous discussions were published e.g. by \cite{dan}, \cite{dan2} . The dataset contains $2167$ fire losses, which occurred between $1980$ and $1990$. For this amount of data we considered bootstrap subsample sizes in the interval $(50,300)$. We calculated the tail index of this dataset with $10\,000$ bootstrap samples for different subsample sizes. The results can be seen in Table \ref{table:danish_subsample}.

\begin{table}[h!]
\centering
\caption{Estimated tail index for Danish fire losses data with different subsample sizes}
\label{table:danish_subsample}
\begin{tabular}{||c| c| c||  }
\hline\hline
 Subsample size &Regression fitting method & Mean regression method  \\
  \hline\hline
50 & 0.689 &0.702 \\
 \hline
100 & 0.687 & 0.68\\
 \hline
 150&0.644 &0.646 \\
 \hline
200 & 0.604 & 0.621\\
\hline
300 & 0.555&0.598\\
\hline\hline

\end{tabular}

\end{table}

For comparison, we calculated the tail index by using other methods. To be more exact we fitted a generalized extreme value model for the maxima of $60$ losses and Pareto distribution over the threshold $u=30$. The chosen values satisfies the "not too high, not too low" rule, which is required for the extreme value modeling. Moreover, we estimated the tail index parameter with the Hill estimation by using double bootstrap and Kolmogorov-Smirnov methods. The results can be seen in Table \ref{table:danish_comp}. Beside that Figure \ref{fig:hillplot} presents the Hill estimation using different $k$ as threshold of ordered statistics.

\begin{figure}[h!]
	\centering
	
	\includegraphics[width=0.6\textwidth]{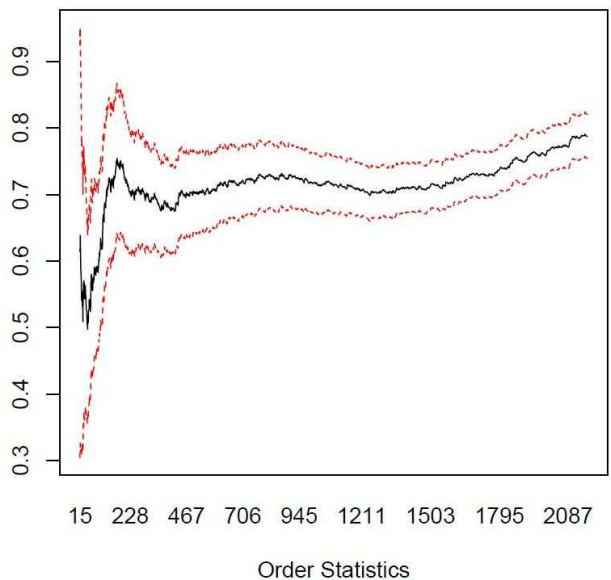} 
	
	\caption{Hill estimation for the tail index parameter of Danish fire losses dataset, using different $k$ values for estimation.}
	\label{fig:hillplot}
\end{figure}

\begin{table}[h]
\centering
\caption{Comparing methods for Danish fire losses data}
\label{table:danish_comp}
\begin{tabular}{||c| c||   }
\hline\hline
 Method & Tail index ($\xi$) \\
 \hline\hline
Fitted GEV & 0.507 \\
 \hline
Fitted GPD & 0.659\\
 \hline
Kolmogorov-Smirnov & 0.61\\
\hline
Double bootstrap & 0.707\\
\hline\hline

\end{tabular}
\end{table}

As one can see, our estimators are between the Kolmogorov-Smirnov and the double bootstrap methods. Using higher subsample size results in lower values. In theory the subsample ratio has to tend to $0$, therefore the estimations from smaller subsample sizes must be closer to the real parameter. 
One can see, that the estimation is near to the result of the double bootstrap. This observation corresponds to the fact that for $\xi>0.5$ the double bootstrap method is more accurate than the one based on the Kolmogorov-Smirnov distance.

With the result of regression fitting method using subsamples of size $150$, we estimated the quantiles of the generalized Pareto distribution with shape parameter $\xi = 0.644$. By setting the above value for the shape and $0$ as the location parameter, the maximum likelihood estimation for scale is $0.95$. 

Table \ref{table:danish_quantile} contains the calculated quantiles of the GPD$(0,0.95,0.644)$.
Moreover we calculated the empirical distribution function for the quantile values. One can see, that the high quantile estimators are close to the observed quantiles of the fire loss data.

 \begin{table}[h]
 \centering
 \caption{Quantiles of generalized Pareto(0,0.95,0.644), empirical distribution for Danish fire losses data at the estimated points, together with the empirical quantiles}
 \label{table:danish_quantile}
 \begin{tabular}{||c|c| c||   }
 \hline\hline
  Quantile & estimated Pareto (emp. dist.)& Empirical quantiles \\
  \hline\hline
 0.95 & 8.68 (0.943) &9.97\\
  \hline
 0.99 & 27.156 (0.99)&26.04\\
  \hline
 0.999 & 124.66 (0.9986)&131.55\\
 \hline\hline
 
 \end{tabular}
 \end{table}

\section{Conclusion}

Our new regression method (\ref{regest}) provides an opportunity to estimate the tail index for heavy-tailed distributions. We have shown its merits by the parameter estimation of known distributions, and presented that our method is also useful for real life data. The computation time is less than using the double bootstrap method. However as our algorithm also applies bootstrap techniques, one can use less bootstrap samples to lower the computation time further if needed -- at the expense of the estimations' accuracy, or conversely.

Our simulations showed in consort of \cite{kol} that the best estimation is based on the Kolmogorov-Smirnov distance (\ref{ksdist}), if the $\xi$ parameter is less than $0.5$. For $0.5 < \xi < 1$ the regression estimation can result the best estimation. However, for higher tail index the methods have a slight dependence of the initial distribution. If the distribution is close to GPD or Student, then the double bootstrap method (\ref{doubleboot}) could result more accurate estimation. In contrast, if the distribution is stable or Frechet we recommend to use the regression estimation. Therefore a fast Hill estimation using the Kolmogorov-Smirnov method and fitting distributions for the sample can help to choose the best method for estimating the tail index. This model selection algorithm could be extended by comparing more methods or by more fitted distributions.

Our regression estimation has two types, using the mean of the bootstrap samples or using the location parameter of the best fitted GEV distribution. Our experiments did not indicate which one is the more precise, therefore we suggest using both in a real life analysis. Further calculations could find the answer for this question.

\section{Acknowledgment}  
       
       The project was supported by the European Union, co-financed by the European Social Fund (EFOP-3.6.3-VEKOP-16-2017-00002).

\end{document}